\documentclass[final]{siamltex}

\usepackage{lastpage} 
\usepackage{fancyhdr}
\usepackage{amsmath}
\usepackage{graphicx}
\usepackage{xcolor}

\newcommand{\norm}[1]{\left|#1\right|}
\newcommand{\vnorm}[1]{\left|\left|#1\right|\right|}

\newtheorem{thm}{Theorem}
\newtheorem{lem}{Lemma}

\title{On the rank-one approximation of symmetric tensors}
\author{Michael J. O'Hara\thanks{Mailbox L-363, Lawrence Livermore National
    Laboratory, 7000 East Ave., Livermore CA 94550 ({\tt
      ohara7@llnl.gov}).}}

\begin{document}

\maketitle

\thispagestyle{fancy}

\begin{abstract}
  The problem of symmetric rank-one approximation of symmetric tensors
  is important in Independent Components Analysis, also known as Blind
  Source Separation, as well as polynomial optimization.  We analyze
  the symmetric rank-one approximation problem for symmetric tensors
  and derive several perturbation results.  Given a symmetric rank-one
  tensor obscured by noise, we provide bounds on the accuracy of the
  best symmetric rank-one approximation for recovering the original
  rank-one structure, and we show that any eigenvector with
  sufficiently large eigenvalue is related to the rank-one structure
  as well. Further, we show that for high-dimensional symmetric
  approximately-rank-one tensors, the generalized Rayleigh quotient is
  mostly close to zero, so the best symmetric rank-one approximation
  corresponds to a prominent global extreme value.  We show that each
  iteration of the Shifted Symmetric Higher Order Power Method
  (SS-HOPM), when applied to a rank-one symmetric tensor, moves
  towards the principal eigenvector for any input and shift parameter,
  under mild conditions. Finally, we explore the best choice of shift
  parameter for SS-HOPM to recover the principal eigenvector.  We show
  that SS-HOPM is guaranteed to converge to an eigenvector of an
  approximately rank-one even-mode tensor for a wider choice of shift
  parameter than it is for a general symmetric tensor.  We also show
  that the principal eigenvector is a stable fixed point of the
  SS-HOPM iteration for a wide range of shift parameters; together
  with a numerical experiment, these results lead to a non-obvious
  recommendation for shift parameter for the symmetric rank-one
  approximation problem.
\end{abstract}



\pagestyle{myheadings}
\thispagestyle{plain}
\markboth{M. J. O'HARA} {SYMMETRIC RANK-ONE APPROXIMATION OF SYMMETRIC TENSORS}

\section{Introduction}

The symmetric rank-one approximation of a symmetric tensor has at
least two important applications.  One is Independent Components
Analysis, known in signals processing as Blind Source Separation
\cite{Lathauwer95, Kofidis01}.  First, we recall classical Principal
Components Analysis (PCA). PCA identifies a basis for a set of random
variables that diagonalizes the covariance matrix, in other words a
basis where the random variables are uncorrelated.  This is necessary
but not sufficient for independence.  A stronger test for independence
is to check whether the off-super-diagonal elements of the four-way
cumulant tensor, a symmetric tensor defined from the fourth-order
statistical moments, are zero.  A linear transformation that achieves
this can be identified by writing the tensor as a sum of symmetric
rank-one terms; one approach uses successive symmetric rank-one
approximations \cite{Wang07}.

Another important application of the symmetric rank-one approximation
of symmetric tensors is in the optimization of a general homogeneous
polynomial over unit length vectors, i.e. the {\em unit sphere}
\cite{Qi09}.  For instance, the symmetric rank-one variant of the
``Time Varying Covariance Approximation 2'' (TVCA2) problem
\cite{Wang10} can be written
\begin{equation}
  \max_{x:\;\vnorm{x}=1} \sum_{t=1}^T (x^T A_t x)^2 \;,
  \label{eqn::TVCA2}
\end{equation}
where $\{ A_t:t=1\dots T\}$ are a given set of covariance matrices,
and the vector norm is the 2-norm (as are all subsequent norms unless
otherwise indicated).  The argument of (\ref{eqn::TVCA2}) is a
degree-4 homogeneous polynomial, and so as we will see the TVCA2
problem can be represented as the best symmetric rank-one
approximation of a symmetric tensor.

Some things are known about the symmetric rank-one approximation
problem.  The best symmetric rank-one approximation in the Frobenius
norm corresponds to the principal tensor eigenvector, and also the
global extreme value of the generalized Rayleigh quotient
\cite{Kofidis02, Lim05}.  It is not clear that these facts help us
solve the symmetric rank-one approximation problem, because tensor
computations are generally notoriously difficult.  For instance it is
known \cite{Hillar09} that (asymmetric) rank-one approximation of a
general mode-3 tensor is NP-complete.  However, there is an algorithm,
the Symmetric Shifted Higher Order Power Method (SS-HOPM)
\cite{Kolda11}, that is guaranteed to find symmetric tensor
eigenvectors.

We address several questions pertaining to the rank-one approximation
of symmetric tensors.  In Section~\ref{sec::structure}, we
  address the structure of approximately-rank-one symmetric
  tensors. A symmetric rank-one tensor obscured with noise
has a best symmetric rank-one approximation that may not be the same
as the original unperturbed tensor; how close is it?  Is only the
principal eigenvector related to the rank-one structure? For a given
symmetric approximately-rank-one tensor, how well-separated is the
principal eigenvalue from the spurious eigenvalues?  In
  Section~\ref{sec::sshopm}, we consider the application of SS-HOPM to
  approximately-rank-one symmetric tensors.  How is the convergence of
  SS-HOPM affected by the approximately-rank-one structure?  When does
  SS-HOPM find the principal eigenvector?  We employ a perturbation
  approach to prove six theorems that provide insight all these
  questions.

\section{Background and notation}

A {\em tensor} is a multi-dimensional array of numbers.  The number of
{\em modes} of the tensor, $m$, is the number of indices required to
specify entries; a mode-2 tensor is a matrix.  The range of
permissible index values $(n_1,\dots n_m)$ are the {\em dimensions} of
the tensor; if all the dimensions are the same, as with symmetric
tensors, we simply write $n$.  A {\em symmetric tensor} has entries
that are invariant under permutation of indices.  For instance, for a
mode-3 symmetric tensor $\mathcal{A}$, we have
$\mathcal{A}_{123}=\mathcal{A}_{231}$.  In this paper, tensors will be
represented with script capital letters, matrices with capital
letters, vectors with lower-case letters, and real numbers with
lowercase Greek letters.  Integers such as indices, dimensions, etc.
will also be lowercase letters (e.g. $m,n,i\dots$).

A {\em symmetric rank-one tensor} is the outer product of a vector
with itself, which we denote using the $\otimes$ operator.  For
instance, given the vector $a$, we can construct a symmetric rank-one
tensor
\begin{equation}
  (\underbrace{a\otimes a\otimes\cdots a}_{\textrm{$m$ times}})
  _{i_1i_2\dots i_m} \equiv (a^{\otimes m})_{i_1i_2\dots i_m} =
  a_{i_1} a_{i_2}\dots a_{i_m} \;.
  \label{eqn::rankone}
\end{equation}
The {\em rank} of a symmetric tensor $\mathcal{A}$ is the fewest
number of symmetric rank-one terms whose sum is $\mathcal{A}$.

Generally, the $m-r$ product of the $m$-mode tensor $\mathcal{A}$ with
the vector $x$ is the $r$-mode tensor defined
\begin{align}
  (\mathcal{A}x^{m-r})_{i_1\dots i_r} = \sum_{i_{r+1},\dots i_m = 1}^n
  \mathcal{A}_{i_1\dots i_m} x_{i_{r+1}}\dots x_{i_m} \;.
  \label{def::tensorproduct}
\end{align}
The special case $r=0$ evaluates to a scalar and, under the constraint
$\vnorm{x}=1$, is called the {\em generalized Rayleigh quotient}
\cite{Zhang01}.  Interestingly, any degree-$m$ homogenous polynomial,
such as (\ref{eqn::TVCA2}), can be written as $\mathcal{A}x^m$ for
some symmetric tensor $\mathcal{A}$ and indeterminate $x$.  In a
miracle of notation, the derivatives are conveniently represented.
The gradient may be written \cite{Kolda11}
\begin{equation}
  \nabla \mathcal{A}x^m = m \mathcal{A}x^{m-1} \;,
  \label{eqn::gradient}
\end{equation}
and the Hessian may be written \cite{Kolda11}
\begin{equation}
  \nabla^2 \mathcal{A}x^m = m(m-1) \mathcal{A}x^{m-2} \;.
  \label{eqn::hessian}
\end{equation}
The problem of maximizing the generalized Rayleigh quotient has the
following Lagrangian:
\begin{equation}
  \mathcal{L}(x, \mu) = \mathcal{A}x^m + \mu (x^Tx - 1) \;,
  \label{eqn::lagrangian}
\end{equation}
where $\mu$ is the Lagrange multiplier. Using (\ref{eqn::gradient}),
we see the critical points of (\ref{eqn::lagrangian}) satisfy the
following symmetric tensor eigenproblem
\begin{equation}
  \mathcal{A}x^{m-1} = \lambda x \;.
  \label{def::zeigen}
\end{equation}
Solutions to (\ref{def::zeigen}) with $\vnorm{x}=1$ are called {\em Z
  eigenvalues and eigenvectors} \cite{Qi05} to distinguish
(\ref{def::zeigen}) from other tensor eigenvector problems, but here
we will simply call them eigenvectors and eigenvalues.  Together, we
call an eigenvector and eigenvalue an {\em eigenpair}.  The {\em
  principal eigenvector/value/pair} is that corresponding to the
largest-magnitude eigenvalue, which may not be unique.  For instance,
if $(x,\lambda)$ is an eigenpair, then if $m$ is even so is
$(-x,\lambda)$, otherwise if $m$ is odd then so is $(-x, -\lambda)$
\cite{Kolda11}.  We will restrict our attention to real solutions to
(\ref{def::zeigen}).

We note that symmetric tensor eigenvectors do not share all the
properties of symmetric matrix eigenvectors, for instance they may not
be orthogonal.  Z eigenvectors are not scale-invariant so limiting our
discussion to normalized eigenvectors is important.  Finally, we note
that because of the relationship between (\ref{def::zeigen}) and
(\ref{eqn::lagrangian}), the principal eigenvector corresponds to the
extreme value of the generalized Rayleigh quotient, and the outer
product of the principal eigenvector with itself, times the principal
eigenvalue, is the best symmetric rank-one approximation of
$\mathcal{A}$ in the Frobenius norm \cite{Kofidis02, Lim05}.

The Shifted Symmetric Higher Order Power Method (SS-HOPM)
\cite{Kolda11}, for a symmetric tensor $\mathcal{A}$, consists of the
iteration
\begin{equation}
  x_{k+1} = \frac{\mathcal{A}x_k^{m-1} + \alpha x_k}
  {\vnorm{\mathcal{A}x_k^{m-1} + \alpha x_k}} \;,
  \label{eqn::shiftedpoweriter}
\end{equation}
where $\alpha$ is a scalar {\em shift parameter}.  An eigenvector $x$
is a stable fixed point of this iteration provided that the Hessian
matrix for (\ref{eqn::shiftedpoweriter}) is positive semidefinite at
$x$.  That condition is known \cite{Kolda11} to be equivalent, for all
$y \perp x$ and $\vnorm{y}=1$, to
\begin{equation}
  \norm{\frac {(m-1)y^T \mathcal{A}x^{m-2} y +\alpha } 
    {\lambda+\alpha}} < 1 \;.
  \label{eqn::stability}
\end{equation}
It is known \cite{Kolda11} that eigenpairs corresponding to local
maxima of the generalized Rayleigh quotient (called {\em negative
  stable} eigenvectors) are stable fixed points of SS-HOPM provided
$\alpha > \beta (\mathcal{A})$, where
\begin{equation}
  \beta(\mathcal{A}) = (m-1) \max_{x:\;\vnorm{x}=1}
  \rho(\mathcal{A}x^{m-2})\;,
  \label{def::beta}
\end{equation}
and $\rho$ returns the spectral radius of a matrix.  Further, it is
known \cite{Kolda11} that if $\alpha > \beta(\mathcal{A})$, the
SS-HOPM iteration monotonically increases the generalized Rayleigh
quotient and converges to a tensor eigenvector.  It is not clear how
to compute $\beta(\mathcal{A})$, but we have the crude bound
\cite{Kolda11}
\begin{equation}
  \beta(\mathcal{A}) \leq \hat{\beta}(\mathcal{A}) =
  (m-1)\sum_{i_1i_2\dots i_m} \norm{\mathcal{A}_{i_1i_2\dots i_m}} \;.
\end{equation}

The following three properties are useful.
\begin{lem}
  For any $n$ dimensional vectors $a$ and $x$, nonnegative integers $m$ and
  $0\leq r \leq m$, the following holds:
  \begin{equation}
    (a^{\otimes m}) x^{m-r} = (a^Tx)^{m-r} a^{\otimes r} \;.
  \end{equation}
  \label{lem::rankoneprod}
\end{lem}
\begin{proof}
  We use (\ref{def::tensorproduct}) and (\ref{eqn::rankone}):
  \begin{align}
    \left((a^{\otimes m}) x^{m-r}\right)_{i_1\dots i_r} 
    &= \sum_{i_{r+1}\dots i_m = 1}^n  (a^{\otimes m})_{i_1,\dots i_m} x_{i_{r+1}}\dots x_{i_m} \\
    &= \sum_{i_{r+1},\dots i_m = 1}^n a_{i_1} a_{i_2}\dots a_{i_m}
    x_{i_{r+1}}\dots x_{i_m}\\
    &= a_{i_1} a_{i_2} \dots a_{i_r} 
    \left(\sum_{i_{r+1} = 1}^n a_{i_{r+1}}x_{i_{r+1}} \right) \dots
    \left(\sum_{i_m = 1}^n a_{i_m}x_{i_m} \right) \\
    &= (a^Tx)^{m-r}   (a^{\otimes r})_{i_1\dots i_r} \;.
  \end{align}
\end{proof}

\begin{lem}[Kolda and Mayo 2011 \cite{Kolda11}]
  For any $m$-mode symmetric tensor $\mathcal{A}$, and any unit-length
  vector $x$,
  \begin{equation}
    \norm{\mathcal{A}x^m}< \frac{\beta(\mathcal{A})}{m-1} \;.
  \end{equation}
  \label{lem::betaprop}
\end{lem}

\begin{lem}
  For any $m$-mode symmetric tensors $\mathcal{A}$ and $\mathcal{B}$, any
  vector $x$, and nonnegative integers $m$ and $0\leq r \leq m$, we have
  \begin{equation}
    (\mathcal{A}+\mathcal{B})x^{m-r} =
    \mathcal{A}x^{m-r}+\mathcal{B}x^{m-r} \;.
  \end{equation}
  \label{lem::distributive}
\end{lem}
Lemma (\ref{lem::distributive}) follows directly from the definition
of tensor-vector multiplication in (\ref{def::tensorproduct}).

\section{Structure of approximately-rank-one symmetric tensors}
\label{sec::structure}

Define
\begin{equation}
  \mathcal{A} = \lambda \cdot a^{\otimes m} + \mathcal{E} \;, 
  \label{eqn::rankoneplusnoise}
\end{equation}
where $a$ is a unit-length $n$ dimensional vector and $\mathcal{E}$ is
a symmetric tensor representing noise.  Clearly if $\mathcal{E}=0$,
then $(a,\lambda)$ is a principal eigenpair, and all unrelated
eigenvalues are zero.  Now let us consider how close is $(a,\lambda)$
to a principal eigenpair when $\mathcal{E}\neq 0$.

\begin{thm}
  Let $\mathcal{A}$ be defined by (\ref{eqn::rankoneplusnoise}).  Then
  a principal eigenvalue $\lambda_p$ obeys
  \begin{equation}
    \norm{\lambda} -\frac{\beta(\mathcal{E})}{m-1} \leq \norm{\lambda_p} \leq
    \norm{\lambda}+\frac{\beta(\mathcal{E})}{m-1} \;,
  \end{equation}
  and the angle $\theta$ between $a$ and
  the corresponding principal eigenvector $x_p$ is
  bounded by
  \begin{equation}
    \norm{\cos^m \theta} \geq 1 - \frac{2\beta(\mathcal{E})}{\norm{\lambda}(m-1)} \;.
  \end{equation}
  \label{thm::boundtheta}
\end{thm}

\begin{proof}
  Since $(x_p, \lambda_p)$ are a tensor eigenpair, and
  $\vnorm{x_p}=1$, we have
  \begin{equation}
    \mathcal{A}x_p^m = \lambda_p \;.
  \end{equation}
  Using Lemma \ref{lem::distributive}, we can write
  \begin{equation}
    \lambda_p = \mathcal{A}x_p^m  = \lambda (a^{\otimes m}) x_p^m + \mathcal{E}x_p^m \;.
  \end{equation}
  Applying Lemma \ref{lem::rankoneprod} we get
  \begin{equation}
    \lambda_p = \lambda (x_p^T a)^m +  \mathcal{E}x_p^m =
    \lambda \cos^m \theta +  \mathcal{E}x_p^m \;,
    \label{eqn::lambda1}
  \end{equation}
  where $\theta$ is the angle between $x_p$ and $a$.  We can use the
  fact that $\norm{\cos\theta} \leq 1$ together with Lemma
  \ref{lem::betaprop} and the triangle inequality to obtain the bound
  \begin{equation}
    \norm{\lambda_p} \leq \norm{\lambda} + \frac{\beta(\mathcal{E})}{m-1} \;.
  \end{equation}
  We also know that $\lambda_p$, as a principal eigenvalue, is a
  largest-magnitude extremum of the generalized Rayleigh quotient.  In
  particular,
  \begin{equation}
    \norm{\lambda_p} \geq \norm{\mathcal{A}a^m} \;.
  \end{equation}
  Now, using Lemmas \ref{lem::rankoneprod},
    \ref{lem::betaprop}, and \ref{lem::distributive}, we get
  \begin{equation}
    \norm{\lambda_p} \geq \norm{\lambda + \mathcal{E}a^m} \geq \norm{\lambda} -
    \frac{\beta(\mathcal{E})}{m-1} \;.
    \label{eqn::lamlb}
  \end{equation}
  This establishes the first part of the theorem.

  Now, we can combine (\ref{eqn::lambda1}) with (\ref{eqn::lamlb}) to
  get
  \begin{equation}
    \norm{\lambda\cos^m \theta +  \mathcal{E}x_p^m}  \geq \norm{\lambda}
    -\frac{\beta(\mathcal{E})}{m-1} \;.
  \end{equation}
  Using Lemma \ref{lem::betaprop} we have
  \begin{equation}
    \norm{\cos^m \theta} \geq 1 - \frac{2\beta(\mathcal{E})}{\norm{\lambda}(m-1)} \;.
  \end{equation}
\end{proof}

Theorem~\ref{thm::boundtheta} means that as $ \beta(\mathcal{E})$
approaches zero, then $x_p$ approaches $a$ or $-a$.  So, if the noise
is small, then the symmetric rank-one approximation of $\mathcal{A}$
corresponding to the principal eigenpair is close to the symmetric
rank-one tensor that we seek.

We would like to find the principal eigenpair.  However, SS-HOPM will
find any eigenvector corresponding to a local maximum of the
generalized Rayleigh quotient (or local minimum, under appropriate
modifications). The following theorem shows that if
$\norm{\mathcal{A}x^m}$ is sufficiently large and $\beta(\mathcal{A})$
is sufficiently small, then $x$ tells us about $a$ even if it is not a
principal eigenvector.

\begin{thm}
  Let $\mathcal{A}$ be defined as in (\ref{eqn::rankoneplusnoise}),
  and assume, for some $x$ so that $\vnorm{x}=1$, we have
  \begin{equation}
    \norm{\mathcal{A}x^m} \geq \epsilon^m +
    \frac{\beta(\mathcal{E})}{m-1} \;,
  \end{equation}
  where $\epsilon > 0$.  Then
  \begin{equation}
    \norm{a^Tx} \geq \epsilon \;.
  \end{equation}
  \label{thm::LargeEqualsClose}
\end{thm}

\begin{proof}
  We have
  \begin{align}
    \norm{\mathcal{A}x^m} &= \norm{(a^Tx)^m + \mathcal{E}x^m} \\
    &\leq \norm{a^Tx}^m + \frac{\beta(\mathcal{E})}{m-1} \;.
  \end{align}
  The proof is by contradiction.   Suppose $\norm{a^Tx}<\epsilon$.   Then
  \begin{equation}
    \norm{\mathcal{A}x^m} < \epsilon^m + \frac{\beta(\mathcal{E})}{m-1} \;.
  \end{equation}
  But this contradicts our assumption.
\end{proof}

Another interesting question is whether the principal eigenvalue is
``well separated'' for an approximately rank-one symmetric tensor.
Unfortunately, we do not know how to characterize the distribution of
the spurious eigenvalues, but we can characterize the distribution of
the function of which they are critical points.

\begin{thm}
  Let $a$ be an $n$-dimensional vector so that $\vnorm{a}=1$.   Let
  $x$ be an $n$-dimensional vector so that $\vnorm{x}=1$, where $x$ is
  drawn randomly from the unit sphere.   Then
  \begin{equation}
    \textrm{Pr}\left(\norm{a^Tx} > \epsilon\right) \leq \frac{1}{n\epsilon^2} \;.
  \end{equation}
  As a consequence, if $\mathcal{A}$ is defined by
  (\ref{eqn::rankoneplusnoise}), then
  \begin{equation}
      \textrm{Pr}\left(\norm{\mathcal{A} x^m}  \geq
        \epsilon^m  + \frac{\beta(\mathcal{E})}{m-1} \right) \leq
      \frac{1}{n\epsilon^2} \;.
      \label{eqn::thm2}
    \end{equation}
  \label{thm::GRQmostlySmall}
\end{thm}

\begin{proof} 
  Because of the rotational symmetry of the uniform distribution on
  the sphere, the distribution of $a^Tx$ is identical to $e_i^Tx$ for
  any $i$, where $e_i$ is a standard basis vector.  In particular, $E
  e_1^Tx = Ea^Tx$ and $\textrm{Var}(e_1^Tx) = \textrm{Var}(a^Tx)$. 
  Evidently $E e_1^Tx = 0$ since $x$ is uniform across the unit
  sphere.  So we can write
  \begin{equation}
    \textrm{Var} (e_1^T x) = E (e_1^Tx) ^2 - (E e_1^T x)^2= E(e_1^T x)^2 \;. 
  \end{equation}
  Next, using the symmetry of the uniform distribution, together with
  the linearity of expectation and the fact that $x$ is unit length, we obtain
  \begin{equation}
    n E(e_1^Tx)^2 = E \sum_{i=1}^n (e_i^Tx)^2 = E 1 = 1 \;.
  \end{equation}
  So $\textrm{Var} (a^T x) = 1/n$.  Using Chebyshev's inequality,
  we can write
  \begin{equation}
    \textrm{Pr}\left(\norm{a^T x}  \geq
      \frac{\kappa}{\sqrt{n}} \right) \leq \frac{1}{\kappa^2} \;.
  \end{equation}
  Let $\epsilon = \kappa / \sqrt{n}$, then
  \begin{equation}
    \textrm{Pr}\left(\norm{a^T x}  \geq \epsilon\right) \leq \frac{1}{n\epsilon^2} \;.
  \end{equation}
  Then (\ref{eqn::thm2}) follows from a direct application of
  Theorem~\ref{thm::LargeEqualsClose}.
\end{proof}

Theorem~\ref{thm::GRQmostlySmall} shows that if $\mathcal{A}$ is
high-dimensional ($n$ is large), then the generalized Rayleigh
quotient is mostly small.  Consequently, the principal eigenpair
should be a prominent extremum of the generalized Rayleigh quotient.

\section{Application of SS-HOPM to approximately-rank-one symmetric tensors}
\label{sec::sshopm}

Let us consider the SS-HOPM method applied to the tensor in
(\ref{eqn::rankoneplusnoise}).  Throughout this section, to
  simplify discussion, we restrict our attention to
  $\lambda>0$. If $m$ is odd, then the eigenvalues come in
pairs $\pm \lambda$, one of which is positive, so at least one
principal eigenpair is a global maximum of the generalized Rayleigh
quotient.  If $m$ is even, then Theorem~\ref{thm::boundtheta} provides
that a principal eigenvector $x_p$ must be close to $-a$ or $a$, which
shows us $\lambda_p \approx \lambda > 0$ so it is also a global
maximum of the generalized Rayleigh quotient.  So, with $\lambda > 0$,
we may restrict our attention to negative stable eigenpairs, namely
those corresponding to maxima of the generalized Rayleigh quotient,
which simplifies discussion of SS-HOPM.

Let us identify a bound on the shift parameter $\alpha$ to guarantee a
given negative-stable eigenpair (those corresponding to local maxima)
of $\mathcal{A}$, as defined in (\ref{eqn::rankoneplusnoise}), is a
stable fixed point of SS-HOPM.
\begin{thm}
  Let $\mathcal{A}$ be defined as in (\ref{eqn::rankoneplusnoise}),
  and $(x_p, \lambda_p)$ be a negative-stable eigenpair.  Let $\theta$
  be the angle between $x$ and $a$.  Then $x$ is a stable fixed point
  for SS-HOPM provided
  \begin{equation}
    \frac{-\lambda_p + (m-1)\lambda\norm{ \sin \theta \cos^{m-2} \theta} +  \beta
      (\mathcal{E})}{2}  <  \alpha \;.
    \label{eqn::alphareq}
  \end{equation}
  \label{thm::boundalpha}
\end{thm}

\begin{proof}
  From (\ref{eqn::stability}), the condition for a stable eigenvector $x_p$
  is, for $y \perp x_p$,
  \begin{equation}
    \norm{\frac {(m-1)y^T \mathcal{A}x_p^{m-2} y +\alpha } 
      {\lambda_p+\alpha}} < 1 \;.
  \end{equation}
  In fact, for negative stable eigenvectors, the expression within the
  norm is always less than one \cite{Kolda11}, and we only need to
  worry about the lower bound
  \begin{equation}
    -1 < \frac {(m-1)y^T \mathcal{A}x_p^{m-2} y +\alpha } 
    {\lambda_p+\alpha}\;.
  \end{equation}
  Applying the definition of $\mathcal{A}$ in
  (\ref{eqn::rankoneplusnoise}), Lemmas
  \ref{lem::distributive} and \ref{lem::rankoneprod}, and the
  definition of $\theta$, we get
  \begin{equation}
    -1 < \frac{ (m-1) y^T \left( \lambda aa^T \cos^{m-2} \theta +
        \mathcal{E} x_p^{m-2}\right) y + \alpha}{\lambda_p +
      \alpha} \;.
    \label{eqn::jacob}
  \end{equation}
Using the fact $y^T x_p x_p^T y = 0$, together
with the properties of canonical angles between subspaces
\cite[p. 43]{Stewart90}, we can write
\begin{align}
    \norm{y^T a a^T y} &= \norm{y^T (aa^T - x_px_p^T) y} \\
    &\leq \sin \theta \;.
  \end{align}
Together with (\ref{def::beta}), we substitute into
(\ref{eqn::jacob}), taking advantage of $\lambda_p+\alpha
  \geq 0$ (required for convergence), to get
\begin{align}
    -1 &<   \frac{-(m-1)\lambda\norm{\sin \theta \cos^{m-2} \theta} -  \beta
      (\mathcal{E}) + \alpha } {\lambda_p + \alpha} \\
    -\lambda_p - \alpha & < -(m-1)\lambda \norm{\sin \theta \cos^{m-2} \theta} -  \beta
    (\mathcal{E}) + \alpha \;,
  \end{align}
  and solving for $\alpha$, we get
  \begin{align}
    \frac{-\lambda_p + (m-1) \lambda \norm{\sin \theta \cos^{m-2} \theta} +  \beta
      (\mathcal{E})}{2}  & <  \alpha \;.
  \end{align}
\end{proof}

In the limit where $\beta(\mathcal{E})$ is small, we know by Theorem
\ref{thm::boundtheta} that $\sin \theta$ is small and, using the
discussion above to address signs, $\lambda_p \approx \lambda$.  So
our requirement simplifies to $-\lambda/2 < \alpha$.  This bound is
much smaller than $\alpha > \hat{\beta}(\mathcal{A})$ provided in
\cite{Kolda11}.  On the other hand, for general eigenvectors where
$\sin \theta$ is not small, but $\beta(\mathcal{E})$ is small, our
requirement simplifies to $\lambda(m/2-1) < \alpha$.  So $\alpha$ in
the range $-\lambda/2 < \alpha < \lambda(m/2-1)$, the positive
principal eigenvector may be a stable fixed point but spurious
eigenvectors may be unstable.

Let us move on to the question of the basin of attraction.  To
simplify the problem, we consider SS-HOPM applied to an unperturbed
rank-one symmetric tensor
\begin{equation}
  \mathcal{A}=\lambda \cdot a^{\otimes m} \;.
  \label{eqn::unperturbedA}
\end{equation}
It is obvious that the unshifted power method, i.e. SS-HOPM with
$\alpha=0$, converges to $a$ from $x$ in one step provided that
$a^Tx\neq 0$, because the ``range'' of the operator
$\mathcal{A}x^{m-1}$ consists only of the vector $a$.  We note that if
$x$ is chosen randomly, $a^Tx \neq 0$ with probability one.  When
$\alpha \neq 0$, convergence is not obvious, but we can
  show that under mild conditions, SS-HOPM moves towards the principal
  eigenvector.

\begin{thm}
  Let $\mathcal{A}$ be defined as in (\ref{eqn::unperturbedA}), with
  $\lambda > 0$.  Let $x_1$ be a vector so that $\vnorm{x_1}=1$, and
  let $\gamma=a^Tx_1$.  Assume $\gamma^{m-2}>0$.  Let $x_2$ be the
  updated vector under SS-HOPM.  Then $\norm{a^Tx_2}>\norm{\gamma}$
  provided
  \begin{equation}
    \alpha > \frac{-\lambda \gamma^{m-2}}{2} \;.
  \end{equation}
  \label{thm::RankOneWorks}
\end{thm}

\begin{proof}
  Let us decompose $x_1$ into its projection onto $a$ and its
  orthogonal component.
  \begin{equation}
    x_1 = \gamma a + \delta x_{a\perp} \;.
  \end{equation}
  Evidently $\gamma^2 + \delta^2 = 1$, and $a^Tx_1 =\gamma$.  From
  (\ref{eqn::shiftedpoweriter}), and using
  Lemma~\ref{lem::rankoneprod}, we have
  \begin{align}
    x_2 &= \frac{\mathcal{A}x_1^{m-1} + \alpha x_1}
    {\vnorm{\mathcal{A}x_1^{m-1} + \alpha x_1}} \\
    &=\frac{\lambda\gamma^{m-1}a + \alpha \gamma a + \alpha \delta x_{a\perp}}
    {\vnorm{\lambda\gamma^{m-1}a + \alpha \gamma a + \alpha \delta
        x_{a\perp}}}\\
    &=\frac{(\lambda\gamma^{m-1} + \alpha \gamma) a + \alpha \delta x_{a\perp}}
    {\sqrt{(\lambda\gamma ^{m-1} + \alpha \gamma)^2 + (\alpha
        \delta)^2}} \;,
  \end{align}
  and so
  \begin{equation}
    a^Tx_2 = \frac{\lambda\gamma ^{m-1} + \alpha \gamma}
    {\sqrt{(\lambda \gamma^{m-1} + \alpha \gamma)^2 + (\alpha \delta)^2}}
    \;.
    \label{eqn::adotx2}
  \end{equation}
  Evidently $\norm{a^Tx_2} > \norm{\gamma}$ is equivalent to
  \begin{align}
    \norm{\frac{\lambda\gamma ^{m-1} + \alpha \gamma}
      {\sqrt{(\lambda\gamma^{m-1} + \alpha \gamma)^2 + (\alpha \delta)^2}} }
    &> \norm{\gamma} \\
    \norm{\frac{\lambda\gamma ^{m-2} + \alpha }
      {\sqrt{(\lambda\gamma^{m-1} + \alpha \gamma)^2 + (\alpha \delta)^2}} }
    &> 1 \\
    \norm{\lambda\gamma ^{m-2} + \alpha } 
    &> \sqrt{(\lambda\gamma^{m-1} + \alpha \gamma)^2 + (\alpha \delta)^2} \\
    (\lambda\gamma ^{m-2} + \alpha ) ^2 
    &> \gamma^2(\lambda\gamma^{m-2} + \alpha )^2 + (\alpha \delta)^2 \\
    (1-\gamma^2) (\lambda\gamma ^{m-2} + \alpha ) ^2 
    &>  (\alpha \delta)^2 \\
    \delta^2 (\lambda\gamma ^{m-2} + \alpha ) ^2 
    &>  (\alpha \delta)^2 \\
    \delta^2 \lambda\gamma^{m-2}(\lambda \gamma^{m-2} + 2 \alpha) &> 0 \;.
  \end{align}
  Now, since $\delta^2>0$, $\lambda>0$, and $\gamma^{m-2}>0$, this
  is equivalent to
  \begin{equation}
    \alpha > \frac{-\lambda \gamma^{m-2}}{2}\;.
  \end{equation}
\end{proof}

Let us discuss the requirement $\gamma^{m-2}>0$.  For $m$ even,
this is true for all $x_1$ given $a^Tx_1 \neq 0$, and so
Theorem~\ref{thm::RankOneWorks} provides that SS-HOPM moves ANY input
vector towards $a$ with probability one.  When $m$ is odd, the
property holds for half of the choices of $x_1$.  However, it is easy
to check using (\ref{eqn::adotx2}) that the sign of $\gamma$ is
preserved under the SS-HOPM update, so repeated applications of
SS-HOPM repeatedly improve $x_i$.

It would be nice to generalize Theorem~\ref{thm::RankOneWorks} to the
case $\mathcal{E}\neq 0$.  However, it cannot hold in the same form
because $a$ is not necessarily a stationary point of SS-HOPM in that
case.  Nonetheless, if the basin of attraction varies smoothly under
small perturbation to the original tensor, then we expect the basin of
attraction for the principal eigenvector to be large for small
$\mathcal{E}$.

We have one more interesting result on the application of SS-HOPM to
approximately-rank-one symmetric tensors, but it only holds for
even-mode tensors.

\begin{thm}
  Let $\mathcal{A}$ be defined as in (\ref{eqn::rankoneplusnoise}),
  and assume $\lambda > 0$, $m$ is even, and the shift parameter
  $\alpha$ for SS-HOPM satisfies $\alpha > \beta(\mathcal{E})$.  Then
  SS-HOPM always increases the generalized Rayleigh quotient and
  converges to an eigenvector.
\end{thm}

\begin{proof}
  Define
  \begin{equation}
    f(x) = \mathcal{A}x^m + (m\alpha/2) (x^T x) \;.
  \end{equation}
  Notice that the second term of $f(x)$ is constant on the unit
  sphere, and the SS-HOPM iteration can be written
  \begin{equation}
    x_{k+1} = \frac{\nabla f(x_k)}{\vnorm{\nabla f(x_k)}} \;.
  \end{equation}
  This iteration is known \cite{Kofidis02, Kolda11} to increase $f(x)$ and
  converge to an eigenvector provided $\nabla^2 f(x)$ is positive
  semidefinite symmetric (PSSD).  We can write
  \begin{align}
    \nabla^2 f(x) &= m(m-1)\mathcal{A}x^{m-2} + m \alpha I \\
   &= m(m-1)\lambda (a^Tx)^{m-2} aa^T +  m(m-1)\mathcal{E}x^{m-2} + m\alpha I \;.
  \end{align}
  Since $\lambda > 0$ and $m$ is even, the first term is PSSD.  So it
  is sufficient to show that the remaining terms 
  \begin{equation}
    m(m-1)\mathcal{E}x^{m-2} + m\alpha I \;.
  \end{equation}
  sum to a PSSD matrix.  But since the last term is merely a spectral
  shift, this is assured provided
  \begin{equation}
    \min_x m \alpha - m(m-1)\rho(\mathcal{E}x^{m-2})) > 0 \;,
  \end{equation}
  which can be written
  \begin{equation}
    \alpha > \beta(\mathcal{E}) \;.
 \end{equation}
\end{proof}

We conducted a numerical experiment that illustrates the
  theorems in this section.  We define a tensor $\mathcal{A}$
  with $n=100$ and $m=4$, and pick $a=(1,0,0\dots)$.  To be
able to use $n$ even this large, we need to define $\mathcal{E}$ as a
sparse tensor.  To generate $\mathcal{E}$, we set $\mathcal{E}=0$,
pick 500 indices at random, and populate those entries with random
Gaussian numbers, zero-mean unit-variance.  We then permute those
indices in all 24 possible ways and copy values to make $\mathcal{E}$
symmetric.  Finally, we scale the elements so that
$\hat{\beta}(\mathcal{E})=0.03$, and so $\hat{\beta}(\mathcal{A})
\approx 3$.

Now, we let $\alpha$ range from $-1$ to $5$, and apply the shifted
power method with 10 random starts.  Let $x$ be the output of the
SS-HOPM, then a success is defined by $\norm{a^Tx} > 0.9$.
Figure~\ref{fig::expt} illustrates the success rate as a function of
$\alpha$.  To compute $\alpha_{min}$ we combine
Theorem~\ref{thm::boundtheta} and Theorem~\ref{thm::boundalpha}, to
get $\alpha_{min} = -0.3365$ for the principal eigenvector and
$\alpha_{min}=1.015$ for the spurious eigenvectors.  Evidently the
best chance of success for converging to the principal eigenvector is
between these two choices of $\alpha$; the fact that the
  success rate can be almost 100\% is supported by
  Theorem~\ref{thm::RankOneWorks}. Choosing $\alpha >
\hat{\beta}(\mathcal{A})$, even though it guarantees the SS-HOPM
iteration increases the generalized Rayleigh quotient and converges,
does not have the best chance of success for recovering the principal
eigenvector.  We speculate that choosing large $\alpha$ results in
more spurious eigenvectors being stable fixed points of the SS-HOPM
iteration, resulting in more spurious answers.

\begin{figure}
  \begin{center}
    \includegraphics[height=2in]{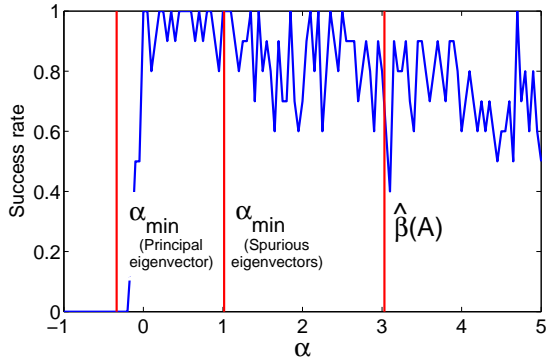}
  \end{center}
  \caption{Success rate for finding the best symmetric rank-one
    approximation of a symmetric tensor, as a function of shift
    parameter $\alpha$.  The values for $\alpha_{min}$ come from
    Theorem~\ref{thm::boundalpha} for the principal and spurious
    eigenvectors.  Recall $\alpha > \alpha_{min}$ is sufficient but
    not necessary for stability.  The best performance for SS-HOPM on
    rank-one approximation is when $\alpha$ is about the
    Theorem~\ref{thm::boundalpha} threshold for the principal
    eigenvector but below the threshold for the spurious
    eigenvectors.}
  \label{fig::expt}
\end{figure}

\section{Conclusion}

Our perturbative analysis establishes new facts about the structure of
approximately-rank-one symmetric tensors, and the application of
SS-HOPM to the rank-one approximation problem. We bound the closeness
of the best symmetric rank-one approximation, and show that any
sufficiently-large eigenpair informs us about the rank-one structure.
We show that in high dimensions, most of the generalized Rayleigh
quotient, whose critical points correspond to eigenvalues, is close to
zero; as a consequence, the principal eigenvalue is prominent.  We
establish that for rank-one symmetric tensors, under mild conditions,
SS-HOPM always moves an input vector towards the principal
eigenvector.  We also show that the principal eigenvector is a stable
fixed point for SS-HOPM under a wide choice of shift parameters, and
that SS-HOPM is guaranteed to converge to an eigenvector for a much
smaller choice of $\alpha$ in the approximately-rank-one case (for an
even number of modes) than the general case.  A complete
characterization of the basin of attraction for the principal
eigenvector remains an open question.  Finally, it is hoped that
better understanding of the symmetric rank-one problem may lead to
better of understanding of more complicated problems such as
Independent Components Analysis.

\section{Acknowledgements}

Thanks to Mark Jacobson, Urmi Holz, Tammy Kolda, Dianne O'Leary, and
Panayot Vassilevski for useful observations and guidance.

\bibliography{foo}

\end{document}